\documentclass[10pt]{article}
\usepackage{amsthm}
\usepackage{algorithm, algorithmic}
\usepackage{graphicx}
\usepackage{amssymb}
\newtheorem{theorem}{Theorem}[section]
\theoremstyle{definition}

\theoremstyle{remark}

\theoremstyle{plain}
\newtheorem{lemma}[theorem]{Lemma}
\newtheorem{corollary}[theorem]{Corollary}

\def\denseformat{
\setlength{\textheight}{9in}
\setlength{\textwidth}{6.9in}
\setlength{\evensidemargin}{-0.2in}
\setlength{\oddsidemargin}{-0.2in}
\setlength{\headsep}{10pt}
\setlength{\topmargin}{-0.3in}
\setlength{\columnsep}{0.375in}
\setlength{\itemsep}{0pt}
}

\denseformat

\usepackage{amsthm}
\usepackage{algorithm, algorithmic}
\usepackage{graphicx}
\usepackage{amssymb}
\def\MathN{\hbox{\rm I\kern-2pt I\kern-3.1pt N}}
\def\Expect{\hbox{\rm I\kern-2pt I\kern-3.1pt E}}
\begin{document}
\title{On the Locality of Some NP-Complete Problems}
\author{
Leonid Barenboim\thanks{Department of Computer Science,
        Ben-Gurion University of the Negev,
        POB 653, Beer-Sheva 84105, Israel.
        E-mail: {\tt leonidba@cs.bgu.ac.il}
        \newline Supported by the Adams Fellowship
Program of the Israel Academy of Sciences and Humanities.}}
\maketitle
\begin{abstract}
We consider the distributed message-passing ${\cal LOCAL}$ model. In this model a communication network is represented by a graph where vertices host processors, and communication is performed over the edges. Computation proceeds in synchronous rounds. The running time of an algorithm is the number of rounds from the beginning until all vertices terminate. Local computation is free. An algorithm is called {\em local} if it terminates within a constant number of rounds. The question of what problems can be computed locally was raised by Naor and Stockmayer  \cite{NS93} in their seminal paper in STOC'93. Since then the quest for problems with local algorithms, and for problems that cannot be computed locally, has become a central research direction in the field of distributed algorithms \cite{KMW04,KMW10,LOW08,PR01}. 
%The currently known problems that can be solved locally have simple sequential algorithms. On the other hand, many problems with simple sequential algorithms cannot be solved locally \cite{KMW04,KMW10}.

We devise the first local algorithm for an {\em NP-complete} problem. Specifically, our randomized algorithm computes, with high probability, an $O(n^{1/2 + \epsilon} \cdot \chi)$-coloring within $O(1)$ rounds, where $\epsilon > 0$ is an arbitrarily small constant, and $\chi$ is the chromatic number of the input graph. (This problem was shown to be NP-complete in \cite{Z07}.) On our way to this result we devise a constant-time algorithm for computing $(O(1), O(n^{1/2 + \epsilon}))$-network-decompositions. Network-decompositions were introduced by Awerbuch et al. \cite{AGLP89}, and are very useful for solving various distributed problems. The best previously-known algorithm for network-decomposition has a polylogarithmic running time (but is applicable for a wider range of parameters) \cite{LS93}. We also devise a $\Delta^{1 + \epsilon}$-coloring algorithm for graphs with sufficiently large maximum degree $\Delta$ that runs within $O(1)$ rounds. It improves the best previously-known result for this family of graphs, which is $O(\log^* n)$ \cite{SW10}.
\end{abstract}

%\spacyformat
%\denseformat
%\begin{document}

\section{Introduction}
%\subsection{Model and Problems}
{1.1 \bf The Model\\}
We consider the distributed message-passing model. This model, widely known as the ${\cal LOCAL}$ model, was formalized by Linial in his seminal paper in FOCS'87 \cite{L87}. In this model a communication network is represented by an $n$-vertex graph $G = (V,E)$ of maximum degree $\Delta = \Delta(G)$.  The vertices of the graph host processors, and communication is performed over the edges. Each vertex has a distinct identity number (henceforth, ID) of size $O(\log n)$ bits. The model is synchronous, meaning that computation proceeds in discrete rounds. In each round vertices are allowed to perform unbounded local computation, and send messages to their neighbors that arrive before the beginning of the next round. The input for a distributed algorithm is the underlying network. However, initially each vertex knows only the number of vertices $n$, and the IDs of its neighbors. Within $r$ rounds, a vertex can learn the topology of its $r$-hop-neighborhood. For a given problem on graphs, a vertex has to compute only its part in the output. For example, for vertex coloring problems, each vertex has to compute only its color. However, the union of outputs of all vertices must constitute a correct solution. The running time of a distributed algorithm is the number of rounds from the beginning until the last vertex terminates. Local computation is free, and is not taken into account. This is motivated by the study of the ability of each vertex to arrive to a solution based on coordination only with close vertices.\\
{\bf 1.2 Problems and Results\\}
A legal vertex coloring is an assignment of colors to vertices, such that each pair of neighbors are assigned distinct colors.
Vertex coloring problems are among the most fundamental and extensively studied problems in the field of distributed algorithms. Many variations have been studied. The most common variation is the $(\Delta + 1)$-coloring problem. The goal of this problem is computing a legal vertex coloring using at most $\Delta + 1$ colors. This problem has a very simple greedy solution in the sequential setting. Specifically, each vertex performs a color selection based on its $1$-hop neighborhood. (The sequential time of the algorithm is linear.) However, in the distributed setting it becomes much more complicated. It is impossible to compute a solution based on an $O(1)$-hop-neighborhood \cite{L87}. The best currently-known deterministic distributed algorithms require $O(\Delta + \log^* n)$ \cite{BE09,K09}, and $2^{O(\sqrt{\log n})}$ \cite{PS95} rounds. The best currently-known randomized algorithm requires $O(\sqrt{\log n} + \log \Delta)$ rounds \cite{SW10}. Moreover, Linial \cite{L87} proved that any distributed algorithm requires $\Omega(\log^* n)$ rounds for computing $(\Delta + 1)$-, and even $\Delta^2$-, coloring. (On the other hand, $O(\Delta^2)$-coloring is currently known to have a distributed algorithm that requires $O(\log^* n)$ rounds \cite{L87}.) On special graph families it is often possible to employ fewer than $\Delta + 1$ colors. However, such algorithms provably cannot terminate within a small number of rounds. In particular, for graphs with arboricity\footnote[1]{The {\em arboricity} is the minimum number of forests that cover the graph edges.} $a$, coloring a graph with $O(a)$ colors can be performed in $O(a^{\epsilon} \cdot \log n)$ rounds \cite{BE09}, for an arbitrarily small constant $\epsilon > 0$. However, any algorithm for this task requires $\Omega(\log n/ \log a)$ rounds \cite{BE08}.

%Since solving such simple sequential problems as $(\Delta + 1)$-coloring becomes very challenging in the distributed setting, it would be natural to conjecture that hard problems require even more rounds. Surprisingly, we show that the opposite is true. 

In the current paper we focus on problems which are hard in the sequential setting, as opposed to the problems mentioned above.
Nevertheless, we show that NP-Complete vertex-coloring problems can be solved within $O(1)$ rounds in the distributed  ${\cal LOCAL}$ setting. We devise a randomized algorithm that computes an $O(n^{1/2 + \epsilon} \chi)$-coloring, within a constant number of rounds, with high probability. ($\chi$ is the chromatic number of the input graph, and $\epsilon >0$ is an arbitrarily small constant.) Computing $O(n^{1 - \epsilon} \chi)$-coloring (and, in particular, $O(n^{1/2 + \epsilon} \chi)$-coloring) is known to be  NP-complete \cite{Z07}. To the best of our knowledge, prior to our work NP-complete problems could be solved only within $O(Diam(G))$ rounds. (In the ${\cal LOCAL}$ model every computable problem can be solved within $O(Diam(G))$ rounds, since all vertices can learn the topology of the entire input graph.)

The question of what problems can be solved within a constant number of rounds is one of the most fundamental questions in the field of distributed algorithms. It was raised around twenty years ago in the seminal paper of Naor and Stockmayer, titled "What can be computed locally?" \cite{NS93}. In this paper, an algorithm that requires $O(1)$ rounds is called {\em a local algorithm}. Despite a very intensive research in this direction that was conducted in the last twenty years, few problems with local algorithms on general graphs are known. (On the other hand, on constant-diameter graphs, any problem can be solved locally. Therefore, in the current setting, this question is meaningful only with respect to families of graphs with superconstant diameter.) Specifically, there are known local algorithms for computing weak-colorings \cite{NS93}, $\Delta$-forests-decomposition \cite{PR01}, edge-defective-colorings \cite{K09}, and dominating-set approximation \cite{KW03,KMW06,LOW08}. (All these problems have simple sequential solutions. In particular, the dominating-set approximation problems for which local algorithms are known can be solved sequentially in polynomial time.) On the other hand, many problems provably cannot be computed locally. In particular, minimum vertex cover, minimum dominating set, maximum independent set, maximum matching, maximal independent set, and maximal matching require $\Omega(\sqrt{\log n})$ rounds \cite{KMW04,KMW10}. (The first three problems are NP-complete. The last three problems have polynomial sequential solutions. In particular, the last two problems have very simple greedy sequential algorithms.) Also, it is known that $\Delta^k$-coloring, for any constant $k$ requires $\Omega(\log^* n)$ rounds. Thus, discovering non-trivial (and, especially, hard) problems that can be computed within $O(1)$ rounds is of significant interest.

In our $O(n^{1/2 + \epsilon} \chi)$-coloring algorithm the vertices perform NP-complete local computations in each round. Therefore, this algorithm has mainly theoretical interest. However, we stress that unless P=NP, it is impossible to solve NP-complete problems by using polynomial local computation per round. Otherwise, we could simulate the network using a single processor. Consequently, we could solve an NP-complete problem sequentially in polynomial time. The sequential running time would be $O(n)$ times the maximal local running time of a processor.  

In addition to $O(n^{1/2 + \epsilon} \chi)$-coloring we devise local algorithms for several problems. 
In these algorithms vertices perform polynomial local computations, and, therefore, may be useful in practice. Specifically, we devise algorithms for $(O(1),O(n^{1/2 + \epsilon}))$-network-decomposition, and for $\Delta^{1+\epsilon}$-coloring graphs with large degree. The best previously known algorithm for network-decomposition requires $O(\log^2 n)$ rounds, but it employs different parameters for the decomposition \cite{LS93}. The best previously known $\Delta^{1 + \epsilon}$-coloring algorithm requires $O(\log^* n)$ time \cite{SW10}. We elaborate on these problems in Section 1.4.\\
%{\bf 1.3 Correlation between the Sequential and Distributed Settings}\\
{\bf 1.3 The Difficulty of Solving NP-Complete Problems Locally}\\
Discovering NP-complete problems that can be solved locally in the distributed setting is interesting for the following reasons. Although the sequential setting and the distributed  ${\cal LOCAL}$ setting are considerably different each from another, it is plausible that many 
%problems are difficult to compute in both settings.
NP-complete problems are also difficult to compute in the distributed setting. 
 Despite that each vertex has an unbounded local computational power, the vertices have limited knowledge about the input graph. 
%In other words, within $r$-hop-rounds each vertex can learn only a subgraph of the input, which is induced by the $r$-hop-neighborhood of the vertex. 
On the other hand, NP-complete problems usually define global constraints, which makes the computation difficult in the occasion of partial input knowledge. Consider, for example, the maximum clique problem. Consider a graph $G$ in which two cliques $K_4$ and $K_3$ are connected by a path of length $\Theta(n)$, and, therefore, are at distance $\Theta(n)$ each from another. All vertices of $K_4$ must decide that they belong to the maximum clique. However, if the vertices of $K_4$ aware only of their $o(n)$-hop-neighborhood, then they cannot distinguish $G$ from another graph $G'$ that connects $K_4$ and $K_5$. Therefore, they cannot always arrive to a correct solution if the number of rounds is $o(n)$. Similar phenomenon occurs in additional problems. 

Another example of 
%correlation between the sequential and the distributed setting 
a difficulty that NP-complete problems arise in the distributed setting can be found in the area of local decision and verification. (See, e.g., \cite{FKP11}). In this area the vertices are required to verify locally the correctness of the output, and at least one vertex needs to react in case of an incorrect output. While some simple sequential problems, such as maximal independent set and maximal matching are locally verifiable, NP-complete problems are more difficult for distributed verification because of their global constraints.\\
{\bf 1.4 Our Techniques\\}
Our main technical contribution is devising {\em network-decomposition} algorithms that require $O(1)$ rounds. Roughly speaking, a network decomposition is a partition of the vertices into clusters of bounded diameter, such that the supergraph formed by contracting clusters into single vertices has bounded chromatic number. (See Section \ref{sc:prelmn} for a formal definition.) Network-decompositions are among the most useful structures in the field of distributed graph algorithms. Once an appropriate network decomposition is computed, it becomes possible to solve efficiently a variety of problems. These problems include vertex colorings, edge colorings, maximal independent set, maximal matching, and additional problems. The best currently-known deterministic $(\Delta + 1)$-coloring algorithms and maximal independent set algorithms employ network-decompositions \cite{PS95}. Since the best currently-known network-decomposition algorithms require quite a large number of rounds (superlogarithmic for deterministic algorithms, and polylogarithmic for randomized ones) the quest for efficient network-decomposition algorithms is of great interest.

We devise a novel partitioning technique that allows computing network-decomposition with cluster diameter $O(1)$ and supergraph chromatic number $O(n^{1/2 + \epsilon})$. Using a randomized algorithm we partition the vertex set of the input graph into subsets. Each subset has its own helpful properties that allow computing the network-decomposition efficiently. Specifically, one of the subsets contains a small dominating set, with high probability. We show that small dominating sets are very useful for computing network-decompositions. Another subset in the partition has bounded maximum degree, with high probability. This is very useful as well, since we can compute a $\Delta^{1 + \epsilon}$-coloring in constant number of rounds on such graphs. Such a coloring is, in particular, a network-decomposition. Once we compute network-decompositions of the subsets, we merge the results to achieve a unified network decomposition of the input graph.\\
{\bf 1.5 Related Work}\\
Cole and Vishkin \cite{CV86} and Goldberg and Plotkin \cite{GP87} devised deterministic $3$-coloring algorithms for paths, cycles and trees that require $O(\log^* n)$ rounds. Luby \cite{L86} and Alon, Babai and Itai \cite{ABI86} devised randomized algorithms for maximal independent set that require $O(\log n)$ rounds. Averbuch, Goldberg, Luby, and Plotkin \cite{AGLP89} devised a deterministic network-decomposition algorithm that requires $2^{O(\sqrt{\log n \log \log n})}$ rounds. It was later improved by Panconesi and Srinivasan \cite{PS95}, who achieved running time of $2^{O(\sqrt{\log n})}$ rounds.
Schneider and Wattenhofer \cite{SW11} devised a randomized coloring algorithm that produces, for a wide range of graphs, a $(1- 1/O(\chi))\Delta$-coloring within $O(\log \chi + \log^* n)$ time. 

To the best of our knowledge, the hardest problem that could be solved locally prior to our work is computing a constant approximation of minimum dominating sets on planar graphs \cite{LOW08}. Although computing minimum dominating sets on planar graphs is NP-complete, the constant approximation for this problem presented in \cite{LOW08} can be computed in polynomial time in the sequential setting.
%Korman, Sereni, and Viennot \cite{KSV11} devised a rather general technique for transforming non-uniform algorithms (that require the knowledge of $n$ and other parameters) into uniform ones (that do not require this knowledge).

\section{Preliminaries} \label{sc:prelmn}
\noindent Unless the base value is specified, all logarithms in this paper are to base 2.\\
The graph $G'=(V',E')$ is a  {\em subgraph}  of $G=(V,E)$, denoted $G' \subseteq G$, if $V' \subseteq V$ and $E' \subseteq E$. For a subset $V' \subseteq V$, the graph $G(V')$ denotes the subgraph of $G$ induced by $V'$.  The {\em degree} of a vertex $v$ in a graph $G = (V,E)$, denoted {\em $deg_G(v)$}, is the number of edges incident to $v$. The {\em distance} between a pair of vertices $u,v \in V$, denoted $dist_G(u,v)$, is the length of the shortest path between $u$ and $v$ in $G$. A vertex $u$ such that $(u,v) \in E$ is called a {\em neighbor} of $v$ in $G$. The {\em neighborhood} of $v$ in $G$, denoted $\Gamma_G(v)$, is the set of neighbors of $v$ in $G$. If the graph $G$ can be understood from context, then we omit the underscript $_G$.  The {\em $r$-hop-neighborhood} of $v$ in $G$ is $\Gamma_G^r(v) = \{u \ | \ dist_G(u,v) \leq r\}$. If the graph $G$ can be understood from context, we use the shortcut $\Gamma_r(v)$ for $\Gamma_G^r(v)$.  The maximum degree of a vertex in $G$, denoted $\Delta(G)$, is defined by $\Delta(G) = \max_{v  \in V} deg(v) $. 
The {\em diameter} of $G$ is the maximum distance between a pair of vertices in $G$.\\
A {\em dominating set} $U \subseteq V$ satisfies that for each $v \in V$, either $v \in U$, or there is a neighbor of $v$ in $U$.
The {\em chromatic number} $\chi(G)$ of a graph $G$ is the minimum number of colors that can be used in a legal coloring of the vertices of $G$. The {\em Minimum-Coloring problem} is the problem of computing a legal coloring of the vertices  of $G$ using $\chi(G)$ colors.
For a graph $G = (V,E)$, a function $f: V \rightarrow \MathN$ is called a {\em label assignment}.
For a graph $G$ with a label assignment $f$, a connected component of vertices with the same label forms a {\em cluster}. More formally, a {\em cluster} is a connected component $U \subseteq V$, such that for each $u,v \in U$, it holds that $f(u) = f(v)$, and for each $u \in U, v \in \Gamma(u) \setminus U$, it holds that $f(u) \neq f(v)$.
A {\em $(d,c)$-network-decomposition} of a graph $G$ is an assignment of labels from the set $\{1,2,...,c\}$ to vertices of $V$, such that each cluster has diameter at most $d$.\\
An algorithm succeeds {\em with high probability} if it succeeds with probability $1 - 1/n^k$, for an arbitrarily large constant $k \geq 1$.
\section{Approximating Minimum-Coloring using Network Decompositions}
In this section we show how to approximate Minimum-Coloring on a graph with a given $(d,c)$-network-decomposition. First, we provide a high-level description of the algorithm. Suppose that we are given a graph $G$, and a label assignment $f: V \rightarrow \{1,2,...,c\}$, such that each cluster has diameter at most $d$. We $c$-approximate Minimum-Coloring in the following way. First, for each cluster $U \subseteq V$, we compute a Minimum-Coloring $\varphi_U: U \rightarrow \{1,2,...,\chi(G(U))\}$. Next, we compute a new color $\varphi(v)$ for each $v \in V$. Let $W$ be the cluster of $v$. (Notice that by definition, each vertex belongs to exactly one cluster.) We set $\varphi(v) = \varphi_W(v) \cdot c + f(v) - 1$. Intuitively, the color $\varphi(w)$ can be seen as the ordered pair $\langle \varphi_W(v), f(v) \rangle$. The coloring $\varphi$ is returned by the algorithm. In the sequel we show that $\varphi$ is a $c$-approximation of Minimum-Coloring of $G$.

Next, we provide a detailed description of a distributed algorithm that employs the high-level idea described above. The algorithm is called {\em Procedure Approximate}. Similarly to all distributed algorithms that we will describe, it defines the behavior of each vertex $v \in V$ in each round. Procedure Approximate accepts as input the label $f(v)$ of $v$, and the number of labels $c$. In the first stage of the procedure, $v$ collects the entire topology of the cluster $W$ that $v$ belongs to. It is widely known (see, e.g., \cite{AGLP89,PS95}) that collecting the topology of an $r$-hop-neighborhood of a vertex $v$ can be performed in $r$ rounds. (We elaborate on this in the Appendix.)
\def\APPg{
In this section we describe how a vertex can collect the topology of its $r$-hop-neighborhood within $r$ rounds. We show this using an inductive argument on $r$.\\ {\bf Base (r = 1):} In the first round, each vertex $v$ knows its neighborhood $\Gamma(v)$.\\ {\bf Step:} Suppose that after $r-1$ rounds, for a positive integer $r$, each vertex knows its $(r-1)$-neighborhood $\Gamma_{r-1}(v)$. Moreover, suppose that it knows for each $u \in \Gamma_{r-1}(v)$ to which vertices in $V$ the vertex $u$ is connected. Then, in round $r$ each vertex sends this information to all its neighbors. Consequently, a vertex $v$ receives the topology of $\Gamma_{r-1}(u)$ from each of its neighbors $u$. Thus, after round $r$, the vertex $v$ knows $\Gamma_{r}(v)$. Moreover, for each vertex $u \in \Gamma_{r}(v)$ it knows the neighbors of $u$ in $V$.} Therefore, each vertex can collect the topology of its $(d + 1)$-neighborhood $\Gamma_{d+1}(v)$ within $(d + 1)$-rounds. Since the diameter of the cluster $W$ of $v$ is at most $d$, it holds that $W \subseteq \Gamma_{d+1}(v)$. Hence $v$ learns the topology of $W$ within $(d+1)$-rounds.

In the second stage, Procedure Approximate computes a Minimum-Coloring of the cluster $W$ of $v$. To this end, it employs a deterministic algorithm that performs exhaustive search locally. Specifically, for $i = 1,2,...,$ the algorithm goes over all possible (either legal or illegal) colorings of $G(W)$ with $i$ colors. For each coloring it checks whether it is legal or not, and terminates in the first time a legal coloring is found. Observe that this technique guarantees that all vertices that belong to the same cluster $W$ compute the same coloring. Indeed, all vertices $w \in W$ have learnt the entire topology of $W$, and perform an exhaustive search on $G(W)$ locally. Since all the vertices perform exactly the same deterministic algorithm that runs on the input $G(W)$, the output is identical for all vertices in $W$. Denote by  $\varphi_W$ the coloring returned by the exhaustive search. Each vertex $v$ sets $\varphi(v) = \varphi_W(v) \cdot c + f(v) - 1$, and terminates. This completes the description of Procedure Approximate. Next we analyze its correctness and running time.

\begin{lemma} \label{dtime}
Procedure Approximate invoked on a graph $G$ with a $(d,c)$-network-decomposition requires $(d + 1)$-rounds.
\end{lemma}
\begin{proof}
The only stage of Procedure Approximate that is not performed locally is the stage that collects the information of a cluster of diameter $d$. This requires $(d + 1)$-rounds.
%\hfill \ensuremath{\Box}
\end{proof}
\begin{lemma} \label{dapr}
Procedure Approximate invoked on a graph $G$ with a $(d,c)$-network-decomposition computes a $c$-approximate Minimum-Coloring of $G$.
\end{lemma}
\begin{proof}
First we prove that Procedure Approximate computes a legal coloring. Let $(u,v)$ be an edge in $E$. We prove that $\varphi(u) \neq \varphi(v)$. If $u$ and $v$ belong to the same cluster $W$, then $f(u) = f(v)$. In this case the exhaustive search computes the coloring $\varphi_W$, such that $\varphi_W(u) \neq \varphi_W(v)$. Therefore,  
$$\varphi(u) = \varphi_W(u) \cdot c + f(u) - 1 \neq \varphi_W(v) \cdot c + f(v) - 1 = \varphi(v).$$   Otherwise, $u$ and $v$ belong to different clusters. Since $u$ and $v$ are neighbors, this implies $f(u) \neq f(v)$. Therefore,  $$\varphi(u) = \varphi_W(u) \cdot c + f(u) - 1 \neq \varphi_W(v) \cdot c + f(v) - 1 = \varphi(v)$$  as well. (If $\varphi_W(u) = \varphi_W(v)$ this is obvious. Otherwise, it holds because  $|\varphi_W(u) \cdot c - \varphi_W(v) \cdot c| \geq c$, and $1 \leq f(u),f(v) \leq c$. Therefore, $$| (\varphi_W(u) \cdot c + f(u) - 1) - (\varphi_W(v) \cdot c + f(v) - 1)| = | \varphi_W(u) \cdot c  - \varphi_W(v) \cdot c + f(u) - f(v)| \geq 1.)$$ 
To summarize, for any pair of neighbors $u,v \in V$, it holds that $\varphi(u) \neq \varphi(v)$. Therefore, $\varphi$ is a legal coloring of $G$.

Next, we prove that the computed coloring is a $c$-approximate Minimum-Coloring of $G$. Observe that for any $W \subseteq V$ it holds that $\chi(G(W)) \leq \chi(G)$. Indeed, a legal coloring of $G$ using $\chi(G)$ colors restricted to $W$ is, in particular, a legal coloring of $G(W)$. Therefore, $G(W)$ can be legally colored with at most $\chi(G)$ colors. Consequently, for each cluster $W \subseteq  V$, and each vertex $v \in W$, it holds that $1 \leq \varphi_W(v) \leq \chi(G(W)) \leq \chi(G)$. Therefore, $\varphi(v) = \varphi_W(v) \cdot c + f(v) - 1 \leq \chi(G) \cdot c + c - 1$. On the other hand, since $\varphi_W(v) \geq 1$ and $f(v) \geq 1$, it holds that $\varphi(v) \geq c$. Therefore, $\varphi$ employs at most $\chi(G) \cdot c + c - 1 - c + 1 = \chi(G) \cdot c$ colors. We remark that all vertices $v \in V$ should subtract $c - 1$ from $\varphi(v)$ to achieve a color in the range $\{1,2,...,\chi(G) \cdot c\}$.
%\hfill \ensuremath{\Box}
\end{proof}
Linial and Saks \cite{LS93}, devised a randomized algorithm for computing $(O(\log n)$, $O(\log n))$-network-decomposition within $O(\log^2 n)$ rounds. If one is willing to spend that much time, then Lemmas \ref{dtime} - \ref{dapr} in conjunction with the algorithm of Linial and Saks allow computing an $O(\log n \cdot \chi (G))$-coloring of an input graph $G$ within $O(\log^2 n)$ rounds.
\begin{corollary}
It is possible to compute an $O(\log n \cdot \chi (G))$-coloring of a graph $G$ within $O(\log^2 n)$-rounds, with high probability. Hence the above algorithm is an $O(\log n)$-approximation for Minimum-Coloring that requires $O(\log^2 n)$ rounds.
\end{corollary}
In the sequel, we devise coloring algorithms that employ more colors, but require a constant number of rounds.
\section{Computing Network Decompositions with Constant Diameter}
{\bf 4.1 Partitioning Procedure\\}
%\subsection{Partitioning Procedure}
In this section we devise an algorithm for computing $(O(1),O(n^{1/2 + \epsilon}))$-network-decompositions, for an arbitrarily small constant $\epsilon > 0$. Using this algorithm in conjunction with Lemmas \ref{dtime} - \ref{dapr} we obtain an $O(n^{1/2 + \epsilon})$-approximation algorithm for Minimum-Coloring of $G$. This algorithm terminates within $O(1)$ rounds. The algorithm for computing $(O(1),O(n^{1/2 + \epsilon}))$-network-decompositions is called {\em Procedure Decompose}. The main idea of the algorithm is partitioning the vertex set $V$ into two subsets $A$ and $B$ that satisfy certain helpful properties. Specifically, the induced subgraph $G(A)$ contains a dominating set $D$ of $A$,  such that the size of $D$ is sufficiently small.  The set $D$ consists of $O(n^{1/2})$ vertices. The set $B$, on the other hand, satisfies a different property. Specifically, the maximum degree of $G(B)$ is bounded by $O(n^{1/2} \log n)$. In the sequel we show how to compute network-decompositions of $G(A)$ and $G(B)$, and how to combine them to achieve the desired network-decomposition of $G$. In this section we devise an algorithm for computing $A$ and $B$ within $O(1)$ rounds.

The algorithm for computing a partition of $V$ into two subsets $A$ and $B$ is called {\em Procedure Partition}. Procedure Partition is a randomized algorithm that works in the following way. Each vertex $v \in V$ holds a local Boolean variable $v_m$. We say that a vertex marks itself if it sets $v_m = \mathit{true}$. The vertex $v$ is unmarked if and only if $v_m = \mathit{false}$. Initially, all vertices are unmarked. The steps that each vertex $v \in V$ performs are described below. \\ \\ \\ \\ 

\begin{enumerate}
\item[] {\bf Procedure Partition}
\item $v$ marks iself with probability $1/n^{1/2}$, independently of other vertices.
\item {\bf if} $v$ is marked, {\bf then} it sends a 'marked' message to all its neighbors.
\item {\bf if} $v$ is marked or $v$ has a marked neighbor, {\bf then} $v$ joins the set $A$. \\
  {\bf else} $v$ joins the set $B$.
\end{enumerate}

Step 1 and 3 of of Procedure Partition are performed locally, and step 2 requires one communication round. Therefore, the procedure requires $O(1)$ rounds. Next, we prove that Procedure Partition partitions $V$ into the subsets $A$ and $B$ that satisfy the properties mentioned above.

\begin{lemma}
The set $A$ contains a dominating set $D$ of $A$, with size $|D| = O(n^{1/2})$, with high probability.
\end{lemma}
\begin{proof}
The set $A$ contains all the vertices that are marked during the execution of the algorithm, and all their neighbors. Denote by $D$ the set of vertices that are marked. The set $D$ is a dominating set of $A$. We show that with high probability, $|D| = O(n^{1/2})$. Let $X_v$ denote the random indicator variable, such that $X_v = 1$ if $v$ marks iself, and $X_v = 0$ otherwise. Let $X = \sum_{v \in V} X_v$ be the sum of  $n$ indicator variables. Let $\gamma > 0$ be an arbitrarily small constant. The expected number of marked vertices is $\Expect (X) = n \cdot 1/n^{1/2} = n^{1/2}$. Hence, by the Chernoff bound for upper tails, it holds that 
$$Pr[X > (1 + \gamma)\Expect (X)] \leq \left (\frac{e^{\gamma}}{(1 + \gamma)^{1 + \gamma}} \right )^{\Expect (X)}$$
Set $\gamma = 1$. It holds that $Pr[X > 2 \Expect (X)] \leq (e/4)^{n^{1/2}} \leq 1/n^k$, for an arbitrarily large constant $k$, and sufficiently large $n$.
%\hfill \ensuremath{\Box}
\end{proof}
\begin{lemma} \label{maxdgr}
The subgraph $G(B)$ induced by $B$ has maximum degree $O(n^{1/2} \log n)$, with high probability.
\end{lemma}
%The proof of Lemma \ref{maxdgr} can be found in Appendix A.\\
%\def\APPb{
\begin{proof}
Let $k$ be an arbitrarily large positive constant. Consider a vertex $v \in V$ such that $deg_G(v) > k \cdot n^{1/2} \log n$. Denote $\delta = deg_G(v)$. Let $y_1,y_2,...,y_{\delta}$, be the neighbors of $v$ in $G$. For $i = 1,2,...,\delta$, let $Y_i$ denote the random indicator variable, such that $Y_i = 1$ if $y_i$ marks itself, and $Y_i = 0$ otherwise. Let $Y = \sum_{i \in [\delta]} Y_i$ be the sum of  $\delta$ indicator variables. Let $\gamma > 0$ be an arbitrarily small constant. The expected number of neighbors of $v$ in $G$ that are marked is $\Expect (Y) = \delta \cdot 1/n^{1/2} \geq k \cdot \log n$. Hence, by the Chernoff bound for lower tails, it holds that 
$$Pr[Y < (1 -\gamma)\Expect (Y)] \leq \left ( \frac{e^{\gamma}}{(1 - \gamma)^{1 - \gamma}} \right)^{\Expect (Y)} < e^{  - \Expect (Y) \cdot (\gamma^2/2)}.$$
Set $\gamma = 1/2$. It holds that $$Pr[Y < 1/2 \cdot \Expect (Y)] <  e^{  - \Expect (Y) \cdot (1/8)} \leq e^{  - k \cdot \log n \cdot (1/8)} < 1/n^{k/8}.$$
Therefore, $Pr[Y = 0] < 1/n^{k/8}$ as well. This probability corresponds to the chances of a given vertex with degree larger than $k \cdot n^{1/2} \log n$ to have all its neighbors unmarked. By the union bound, the probability that there exists a vertex $v \in V$ with $deg_G(v) > k \cdot n^{1/2} \log n$, such that all neighbors of $v$ in $G$ are unmarked is at most $\rho = n \cdot 1/n^{k/8} = 1/n^{k/8 - 1}$. Hence, with probability at least  $1 - \rho$, all vertices with degree greater than $k \cdot n^{1/2} \log n$ in $G$ have a marked neighbor, and thus join the set $A$. Therefore, all vertices that join $B$ have degree at most  $k \cdot n^{1/2} \log n = O(n^{1/2} \log n)$, with probability at least $1 -1/n^{k/8 - 1}$. Since $k$ is an arbitrarily large constant, the claim in the lemma follows.
%\hfill \ensuremath{\Box}
\end{proof}
%}
{\bf 4.2 Network-decompositions in graphs with bounded degree\\}
%\subsection{Network-decompositions in graphs with bounded degree}
In this section we device an algorithm that allows computing an $(O(1), n^{1/2 + \epsilon})$-network-decomposition of $B$. (We postpone the description of the algorithm for $A$ to Section 4.3.) The algorithm we devise to be used for $B$ is quite general. It computes a legal coloring of the underlying graph, rather than a network-decomposition. However, a legal coloring using $\ell$ colors is, in particular, an $(O(1), \ell)$-network-decomposition. Moreover, our algorithm colors any graph with maximum degree $\Delta \geq n^{\mu}$, using $\Delta^{1 + \epsilon}$ colors, for arbitrarily small constants $\epsilon, \mu > 0$. For graphs with maximum degree smaller than $n^{\mu}$, our algorithm produces an $O(n^{\mu + \epsilon \cdot \mu})$-coloring. Observe that applying this algorithm on $G(B)$ results in an $O((n^{1/2} \log n)^{1 + \epsilon})$-coloring of $G(B)$, which is an $O(n^{1/2 + \epsilon})$-coloring.

The algorithm is called {\em Procedure Color}. It accepts as input a parameter $\Delta$ which is an upper bound of the maximum degree of the underlying graph, such that $\Delta \geq n^{\mu}$. The procedure performs a constant number of rounds. Each round consists of the following steps.

\begin{enumerate}
\item[] {\bf Procedure Color}
\item $v$ draws uniformly at random a color $q_v$ from the range $\{1,2,...,\left \lceil \Delta^{1 + \epsilon} \right \rceil \}$, and sends $q_v$ to all neighbors.
\item {\bf if} $q_v$ is different from the colors of all neighbors of $v$ (including those that have already terminated) \\ {\bf then} $v$ sets $q_v$ as its final color, informs its neighbors, and terminates.
\item {\bf else} $v$ discards the color $q_v$.
\end{enumerate}
Observe that once the procedure terminates in all vertices, each vertex $v$ holds a color $q_v$ that is different from the colors of all its neighbors. To prove this, let $i$ denote the round in which a vertex $v$ has terminated. All neighbors $u$ of $v$ that have terminated before $v$, have selected a final color $q_u$ before round $i$. Thus, these colors do not change in round $i$ and afterwards. Therefore, $q_v \neq q_u$ for each neighbor $u$ of $v$ that has terminated before $v$. (Otherwise, $v$ would not terminate in round $i$.) By the same argument, we conclude that for each neighbor $u$ that terminates after $v$, it holds that $q_u \neq q_v$. It is left to show that for all neighbors $w$ of $v$ that terminates in round $i$, it holds that $q_v \neq q_w$. Assume for contradiction that the colors that $v$ and $w$ select in round $i$ are identical. Then nor $v$ nor $w$ terminate in round $i$. This is a contradiction. Therefore, if all vertices terminate, the produced coloring is legal.

Next, we analyze the performance of Procedure Color in case that it is executed for a single round. (Not necessarily the first one.)
\begin{lemma} \label{colr}
Suppose that Procedure Color is executed in round $i$, for $i \geq 1$, by a vertex $v$. The probability that $v$ does not terminate in round $i$ is at most $1/\Delta^{\epsilon}$.
\end{lemma}
\begin{proof}
Assume without loss of generality that in round $i$ the vertex $v$ selects a color after all its neighbors do so.
The number of different colors selected by all neighbors of $v$ is at most $\Delta$. The vertex $v$ selects a color from the range $\{1,2,...,\left \lceil \Delta^{1 + \epsilon} \right \rceil \}$. Therefore, the probability that it selects a color that is identical to a color of a neighbor is at most $\Delta/\Delta^{1 + \epsilon} = 1/\Delta^{\epsilon}$.
%\hfill \ensuremath{\Box}
\end{proof}
Next, we analyze the probability that a vertex does not terminate within $i$ rounds, for an integer constant $i > 0$.
The probability that a vertex does not terminate in the first round is at most $1/\Delta^{\epsilon}$, by Lemma \ref{colr}.
 The probability that a vertex does not terminates in round $i$, conditioned on that it does not terminate within rounds $1,2,...,i-1$, is at most $1/\Delta^{\epsilon}$ as well. Therefore, the probability that a vertex does not terminate within $i$ rounds is $\hat{\rho} = (1/\Delta^{\epsilon})^i$. For an arbitrarily large constant $k$, set $i = \left \lceil k/(\mu \cdot \epsilon) \right \rceil$. It holds that $\hat{\rho} = (1/\Delta^{\epsilon})^i \leq (1/n^{\mu \cdot \epsilon})^i \leq 1/n^k$. By the union bound, the probability that there exists a vertex that does not terminate is at most $n \cdot 1/n^k = 1/n^{k - 1}$. Therefore, all vertices terminate with high probability. We summarize this in the following theorem.
 
\begin{theorem}
After a constant number of rounds, Procedure Color computes a legal $\Delta^{1 + \epsilon}$-coloring of an input graph of maximum degree at most $\Delta \geq n^{\mu}$, for arbitrarily small constants $\epsilon, \mu > 0$, with high probability.
\end{theorem}

\begin{corollary} \label{degr}
For a graph with maximum degree $O(n^{1/2} \log n)$, an $(O(1), n^{1/2 + \epsilon})$-network-decomposition can be computed in $O(1)$ rounds, with high probability. \\
\end{corollary}
{\bf 4.3 Network-decompositions in graphs with small dominating set \\}
%\subsection{Network-decompositions in graphs with small dominating set} \label{sc:set}
In this section we devise an algorithm for computing $(O(1), n^{1/2 + \epsilon})$-network-decomposition for graph that contain a dominating set of size $O(n^{1/2})$. We remark that it is possible to obtain a more general variant of the algorithm, that computes $(O(1),n^{\mu + \epsilon})$-network-decomposition for graph with dominating sets of size $O(n^{\mu})$, for a constant $0 < \mu < 1$, and an arbitrarily small constant $\epsilon > 0$. For clarity, and because it is sufficient for our goals, we present here only the algorithm for graphs with dominating sets of size  $O(n^{1/2})$.
This algorithm can be used for $A$. The algorithm is called {\em Procedure Dominate}. It accepts as input a dominating set $D$ of the underlying graph $G' = (V',E')$, such that $D$ contains $O(n^{1/2})$ vertices. Our ultimate goal would be assigning unique labels from the range $\{1,2,...,O(n^{1/2})\}$, to the vertices of $D$. If we would be able to do so, then each vertex in $V' \setminus D$ could select a label of (an arbitrary) neighbor that belongs to $D$. Consequently, the diameter of each cluster would be at most $2$. Indeed, all vertices with the same label in $V' \setminus D$ are connected to a common vertex in $D$, and all vertices in $D$ have distinct labels. 

It is currently unknown whether it is possible to compute labels for vertices in $D$ as described above, within a constant number of rounds. We address a problem with somewhat weaker requirements. This is, however, sufficient for computing the desired network-decomposition. Specifically, we require that the labels assigned to vertices of $D$ are taken from the range $\{1,2,...,O(n^{1/2 + \epsilon})\}$. Also, we do not require that all the labels are unique. Instead, we require that each vertex $v \in D$ selects a label that is distinct from the labels of vertices in $\Gamma_{G'}^3(v) \cap D$. In particular, such a labeling constitutes a distance-3 coloring of $D$. This way, once vertices from $D$ select appropriate labels, and vertices from $V' \setminus D$ select a label of an arbitrary neighbor from $D$, we obtain an $(O(1),n^{1/2 + \epsilon})$-network-decomposition. We will prove this claim shortly, but first we describe the algorithm in a more detail. 
In each iteration, each vertex $v \in D$ performs the following steps. These steps are performed for $k$ iterations, where $k > 0$ is an integer constant to be determined later. The vertices of $V' \setminus D$ do not perform any steps in these $k$ iterations.

\begin{enumerate}
\item[] {\bf Procedure Dominate}
\item[] (Performed by vertices $v \in D$)
\item $v$ draws uniformly at random a label $l_v$ from the set $\{1,2,...,\left \lfloor n^{1/2 + \epsilon} \right \rfloor \}$.
\item $v$ collects the topology of $\Gamma_{G'}^3(v)$, including labels.
\item {\bf if} $l_v$ is distinct from all labels in $\Gamma_{G'}^3(v) \cap D$ \\ {\bf then} 
$v$ sets $l_v$ as its final label, informs its neighbors, and terminates.
\item {\bf else} $v$ discards the label $l_v$.
\end{enumerate}
Observe that each iteration of Procedure Dominate requires four rounds. Three rounds are required for collecting the topology of $\Gamma_{G'}^3(v)$, and one round is required for informing the neighbors about a final selection of $l_v$. We will prove that after $k$ iterations, for a sufficiently large constant $k$, all vertices $v \in D$ terminate, with high probability. In iteration $k + 1$, all vertices $u \in V' \setminus D$ select a final label $l_u$, such that $l_u = l_w$, for an arbitrary $w \in \Gamma_{G'}(v) \cap D$. This completes the description of Procedure Dominate. We analyze its correctness and running time below.

\begin{lemma}
After a constant number of iterations of executing Procedure Dominate, all vertices $v \in D$ terminate, with high probability.
\end{lemma}
\begin{proof}
Observe that for each $v \in D$, the number of vertices in $\Gamma_{G'}^3(v) \cap D$ is $O(n^{1/2})$, since $|D| = O(n^{1/2})$. Therefore, the probability that a vertex does not terminate after a single iteration is $O(n^{1/2})/n^{1/2 + \epsilon} = 1/\Omega(n^{\epsilon})$. The probability that a vertex does not terminate after $k$ iterations is $1/\Omega(n^{\epsilon \cdot k})$. For an arbitrarily large constant $k'$, there exist a sufficiently large constant $k$, such that $1/\Omega(n^{\epsilon \cdot k}) < 1/n^{k'}$. Therefore, a vertex terminates after a constant number of iterations, with high probability. By the union bound, all vertices in $D$ terminate within $k'$ iterations, with probability at least $1 - 1/n^{k' - 1}$.
%\hfill \ensuremath{\Box}
\end{proof}
\begin{lemma} \label{domint}
Procedure Dominate computes an $(O(1),O(n^{1/2 + \epsilon}))$-network- decomposition, with high probability.
\end{lemma}
%The proof of Lemma \ref{domint} can be found in Appendix A.
%\def\APPd{
\begin{proof}
First observe that all labels are taken from the range $\{1,2,...,\left \lfloor n^{1/2 + \epsilon} \right \rfloor\}$. Next, we show that if all vertices of $D$ terminate, then all clusters have diameter at most $2$. Consider a pair of vertices $u$ and $v$ that belong to the same cluster. There exist a path that connects $u$ and $v$ in which all vertices have the same labels $l_u = l_v$. Let $P = \{w_1,w_2,...,w_{\ell}\}$ be the shortest path among all such paths. It holds that $w_1 = u$ and $w_{\ell} = v$. The length of $P$ is $\ell$. Suppose for contradiction that $\ell > 3$. One of the following three possibilities hold: {\bf (a)} $w_1 \in D$, {\bf (b)} $w_2 \in D$, or {\bf (c)} there is a common neighbor $y \in D$ of $w_1$ and $w_2$. (Otherwise, $w_1$ and $w_2$ do not have common neighbors in $D$. Therefore, they select labels of distinct vertices in $D$ that are at distance at most 3 each from another. Thus, $w_1$ and $w_2$ cannot have the same labels.)\\
{\bf (a)} If  $w_1 \in D$, then we consider the vertex $w_3$. Obviously, $w_3 \notin D$, since otherwise $w_3 \in \Gamma_{G'}^3(w_1) \cap D$, and $w_3$ cannot have the same label as that of $w_1$. Therefore, $w_3$ has a neighbor $z$ that belongs to $D$ and has a label identical to the label of $w_3$. (But $z \neq w_1$ since $P$ is a shortest path.) It holds that $l_z = l_{w_3} = l_{w_1}$. On the other hand, $z \in \Gamma_{G'}^3(w_1) \cap D$ (because $w_3$ is at distance $2$ from $w_1$, and $z$ is a neighbor of $w_3$). Therefore, $l_{w_1}$ cannot be identical to $l_z$. This is a contradiction. \\
{\bf (b)} If  $w_2 \in D$, then we consider the vertex $w_4$. Obviously, $w_4 \notin D$. Therefore, $w_4$ has a neighbor $z$ that belongs to $D$ and has a label identical to the label of $w_4$. It holds that $l_z = l_{w_2}$. On the other hand, $z \in \Gamma_{G'}^3(w_2) \cap D$. Therefore, $l_z$ cannot be identical to $l_{w_2}$. This is a contradiction. \\
{\bf (c)} If there is a common neighbor $y \in D$ of $w_1$ and $w_2$, then $w_1$ and $w_2$ must select a label of a common neighbor. (Otherwise, they would select labels of distinct vertices that are at distance at most 3 each from another, and, therefore, their labels would be distinct.) Assume without loss of generality that this common neighbor is $y$. It holds that $(y,w4) \notin E'$. (Otherwise, $P$ would not be a shortest path.) The vertex $w_3$ may or may not be connected to $y$. Let $w_i \in \{w_3,w_4\}$ be the vertex with the smallest index $i$ such that $(y,w_i) \notin E'$. The distance between $y$ and $w_i$ is $2$. Thus $w_i \notin D$, but has a neighbor $z \in D$, such that $l_y = l_{w_i} = l_z$. But $y \neq z$, and $y,z$ are at distance at most $3$ each from another, thus cannot have the same label. This is a contradiction.\\
Therefore, the assumption that $\ell > 3$ leads to contradictions in all possible cases. Hence $\ell \leq 3$. Therefore, all clusters have diameter at most $2$ as required.
%\hfill \ensuremath{\Box}
\end{proof}
%}
Using Lemma \ref{domint} we can compute an $(O(1),O(n^{1/2 + \epsilon}))$-network-decomposition of $G(A)$. Using Corollary \ref{degr} we can compute an $(O(1),O(n^{1/2 + \epsilon}))$-network-decomposition of $G(B)$. It is left to show how to combine these two network-decompositions to obtain a unified network-decomposition of $G$. To this end, each vertex $u \in A$, with a label $l_u$, computes a new label $l'_u = l_u \cdot 2$. Each vertex $v \in B$, with a label $l_v$, computes a new label $l'_v = l_v \cdot 2 + 1$. Consequently, for each $v \in A, u \in B$, it holds that $l'_v \neq l'_u$. Consider a cluster in $G$ with respect to the new labeling. All the vertices in the cluster have the same new label. Hence all of them have the same old label, as well. Therefore, either all of them belong to $A$, or all of them belong to $B$. Therefore, the diameter of the cluster is $O(1)$. Since the number of labels in $A$ and in $B$ is $O(n^{1/2 + \epsilon})$ the total number of new labels is $O(n^{1/2 + \epsilon})$ as well.

Finally, observe that all the procedures can be combined to produce an $O(n^{1/2 + \epsilon} \cdot \chi(G))$-coloring of $G$ from scratch. To this end, the vertices first compute the value $t = \left \lfloor k \cdot n^{1/2} \cdot \log n \right \rfloor$, where $k$ is the constant hidden in the $O$-notation in Lemma \ref{maxdgr}. (Recall that all vertices know $n$.) Then, they invoke Procedure Partition. Consequently, each vertex knows whether it belongs to $A$ or to $B$. Moreover, the vertices in $A$ know whether they belong to the dominating set $D$ or not. (The vertices that belong to $D$ are marked.) Next, the vertices of $A$ execute Procedure Dominate. The vertices of $B$ execute Procedure Color with the value $t$ as input. Consequently, the desired network-decompositions of $G(A)$ and $G(B)$ are computed. Then they are combined into a unified $(O(1),O(n^{1/2 + \epsilon})$-decomposition of $G$. Next, an $O(n^{1/2 + \epsilon} \cdot \chi(G))$-coloring is computed using Procedure Approximate in $O(1)$ rounds. (See Lemmas \ref{dtime} -\ref{dapr}.) The input for Procedure Approximate, $t' = O(n^{1/2 + \epsilon})$, can be computed locally by each vertex. We summarize this discussion in the following theorem.
\begin{theorem}
For any graph $G$ with $n$ vertices, and an arbitrarily small positive constant $\epsilon >0 $, with high probability, we can compute within $O(1)$ rounds: \\
{\bf (1)} An $(O(1),O(n^{1/2 + \epsilon})$-network-decomposition of $G$. \\
{\bf (2)} An $O(n^{1/2 + \epsilon} \cdot \chi(G))$-coloring of $G$.
\end{theorem}

\noindent {\large {\bf Acknowledgements}}\\
The author is grateful to Michael Elkin for fruitful discussions and very helpful remarks.

\clearpage
\appendix
\centerline{\LARGE\bf Appendix}
%\section{Some proofs}
%\noindent \textbf{Proof of Lemma 4.2:}\\
%\APPb

%\noindent \textbf{Proof of Lemma 4.7:}\\
%\APPd

\section{Collecting neighborhood topology}
\APPg
\end{document}